\documentclass[runningheads]{article}
\usepackage{graphicx} 
\usepackage{amsmath}
\usepackage{amssymb}
\usepackage{amsthm}
\usepackage{authblk}
\usepackage{comment}

\newtheorem{definition}{Definition}

\newtheorem{axiom}{Axiom}
\newtheorem{proposition}{Proposition}
\newtheorem{example}{Example}
\newtheorem{theorem}{Theorem}
\newtheorem{lemma}{Lemma}

\usepackage{geometry}
\usepackage[textsize=tiny]{todonotes}

\title{EIP-4844 Economics and Rollup Strategies}

\author[1]{Davide Crapis}
\author[2]{Edward W. Felten}
\author[2]{Akaki Mamageishvili}

\affil[1]{Ethereum Foundation}
\affil[2]{Offchain Labs}

\date{September, 2023}

\begin{document}

\maketitle

\begin{abstract}
    We study the economics of the Ethereum improvement proposal 4844 and its effect on rollups' data posting strategies. Rollups' cost consists of two parts: data posting and delay. In the new proposal, the data posting cost corresponds to a blob posting cost and is fixed in each block, no matter how much of the blob is utilized by the rollup. The tradeoff is clear: the rollup prefers to post a full blob, but if its transaction arrival rate is low, filling up a blob space causes too large delay cost. The first result of the paper shows that if a rollup transaction arrival rate is too low, it prefers to use the regular blockspace market for data posting, as it offers a more flexible cost structure. Second, we show that shared blob posting is not always beneficial for participating rollups and change in the aggregate blob posting cost in the equilibrium depends on the types of participating rollups. In the end, we discuss blob cost-sharing rules from an axiomatic angle. 
\end{abstract}

\section{Introduction}
Ethereum improvement proposal (EIP) numbered 4844,~{dubbed as EIP-4844}, is meant to create a cheaper and more efficient calldata posting service on the Ethereum main chain, sometimes called layer one (L1). The goal is to facilitate the Ethereum ecosystem move to rollups. Today, the largest optimistic and ZK rollups offer fees that are~3-50x lower than Ethereum L1. This EIP and follow-ups will further reduce costs of transacting on rollups by providing extra space, thus creating strong incentives for users to switch to using rollups and enabling new applications that can borrow Ethereum L1 security at a much lower cost.

In this project, we study the economics of the proposal. We refer to the market created by EIP-4844 as the data market and the gas market of the Ethereum mainnet as the main market. In particular, we look at the trade-offs faced by rollups that are adopting the new service: 

\begin{itemize}
    \item[(1)] When should a rollup use the data market versus the main market for sending data to L1?
    \item[(2)] Is there a substantial efficiency gain in aggregating data from multiple rollups and what happens to the data market fees?
    \item[(3)] When would rollups decide to aggregate and what is the optimal cost-sharing scheme?
\end{itemize}

In what follows we set up an economic model of aggregate rollup demand that we use to study the above questions. We make simplifying assumptions that allow us to obtain a crisp characterization of optimal rollup data posting strategies. In particular, we consider a continuous-time model in a large market with many rollups. 

We model the cost of rollups as a sum of two parts. The first part of the cost is observable, it is a data posting cost. The second is delay cost. For some applications delay of L1 finality is not crucial. However, many applications built on top of rollups need fast L1 finality for their liveness and some even include it in their security model. The goal of the rollup is to minimize overall costs per transaction, as the per-transaction costs are what users incur when using rollups. The blob posting cost is calculated in the equilibrium state endogenously. The main market price per gas is assumed to be fixed for simplicity. The rollups decide between using either of the two technologies for transaction-relevant data posting, that is, both markets are perfect substitutes for each other~\footnote{Using both markets interchangeably is technologically feasible.}.  We show that if the demand for blob posting is high, it drives smaller rollups to use the main market data posting strategy. We identify conditions when posting data on the main market is better for the rollup than posting blobs. In the second part of the paper, the joint blob posting option for two (or more) rollups is studied. Depending on the type of rollups that decide to post a shared blob, the price of the blob in the equilibrium can change in both directions, up or down. We derive bounds on the increase or decrease of the blob price in relative terms. If the shared blob posting is profitable, the rollups can join in doing so. We study cost-sharing rules from the axiomatic angle, by employing the Nash bargaining solution concept from the economic theory. In particular, we show that if both rollups were using data market, in the joint blob posting large rollup has to pay less than a proportional cost of the joint blob. On the other hand, it always pays more than half of the joint blob cost and the improvement of the large rollup is always less than the improvement of the small rollup.

\subsection{Related Literature}

In~\cite{opt_chaos}, the authors study how under the EIP-1559 dynamic fee mechanism the gas usage converges to the target usage over time. This theoretical and empirical observation is a cornerstone of our modeling of the equilibrium state.~\cite{dynamic_posted_price} proposes a dynamic posted-price mechanism—which uses not only block utilization
but also observable bids from past blocks to compute a posted price for subsequent blocks, they study the stability and welfare optimality of the proposed policy in a steady state.
In~\cite{multi_dim_pricing}, the authors argue that different resource types should be independently priced, as converting all resources in one unit and pricing it uniformly is economically inefficient. EIP-4844 can be seen as the first step towards fixing the economic inefficiencies of pricing different types of resources together.~\cite{crapis2023eip4844} provides a definition and empirical analysis of the upcoming EIP-4844 fee market, which introduces data gas on the Ethereum blockchain and a data gas pricing mechanism modeled after the fee update mechanism of EIP-1559.~\cite{multidimensional} studies optimal dynamic pricing method for multiple resources, with uncertain future demand flow and statistical dependencies between resource demands. In~\cite{eff_batch_posting} the authors provide a rollup data batch posting strategy, in the context of a single independent rollup, and in the absence of a dedicated market for data. They analyze trade-offs between price and time that are similar in nature to the ones faced here in the presence of a dual market and multiple rollups.

\section{Continuous-time Model}

In this section, we outline modeling assumptions and derive initial results on the rollup behavior in the presence of two potential markets. We make the following list of assumptions:

\begin{itemize}
\item[$\bullet$] Delay cost of a transaction is $aD$, where $D$ is the delay, in time units, experienced by the transaction (a time when a transaction was posted in a batch or blob to L1, minus time when it arrived to L2 sequencer), and $a>0$ is a positive constant.
\item[$\bullet$] L1 gas price is $G$, which is treated as a constant.
\item[$\bullet$] The base cost of a batch- or blob-posting transaction on L1 is $P_0G$. Here $P_0$ indicates the size of the metadata associated with the rollup transaction containing a batch/blob.
\item[$\bullet$] The cost of posting a blob on L1 is $P_0G+B$, where $B$ will be set later to a market-clearing price. The latter is interpreted as the minimum price for which no more than three blobs are posted per time unit on average.
\item[$\bullet$] The cost of posting a batch of $n$ transactions on L1 is $(P_0+P_1n)G$.
\end{itemize}

The target number of blobs per Ethereum block is denoted by $k$~\footnote{Initially set to $2$ and currently set to $3$.}.
We treat time as continuous so that a blob can be posted at any time. Conceptually, a ``time unit" can be thought of as one L1 block time, but in this model, we will allow blobs to be posted ``in between" L1 blocks.
Suppose there are rollups with transaction arrival rates $R_1,R_2,..., R_n$ and they are sorted in decreasing rates, $R_{i}\geq R_{i+1}$ for any $i\in \{1,...,n-1\}$. The goal of a rollup is to minimize cost per transaction. The latter is obtained by dividing the total cost by the total number of transactions in the blob and paid by rollup users. The justification of this approach is that even though some transactions arrive earlier than others, their arrival time can be assumed to be uniformly random. This payment rule corresponds to the average costs, and therefore, the users pay fairly.

\section{Analysis} 

In this section, we obtain results on the tradeoff between using data and main markets. First, we show the following:

\begin{proposition}\label{threshold_result}
There is a threshold $n^*$, such that for any $i>n^*$, a rollup $i$ uses the direct on L1 posting strategy.
\end{proposition}

\begin{proof}
Let us determine the optimal strategy for a single rollup, assuming the prices on both the data market and on L1 are fixed.
There are two strategies any rollup could use: (1) posting blobs or (2) posting directly on L1. The rollup finds a strategy that minimizes the cost within each of these strategies separately and then selects the strategy that has a lower cost. 

First, consider a blob-posting strategy. Suppose a rollup with generic rate $R$ posts a blob every time $t$. Then, a blob contains $Rt$ transactions, with the total cost (posting cost plus delay cost) of

\begin{equation}\label{total_blob_cost_t}
    T_B(t):=P_0G+B+\frac{aRt^2}{2}.
\end{equation}

The cost per transaction is: 

\begin{equation}\label{per_transaction_blob}
Tr_B(t):=\frac{T_B(t)}{Rt} = \frac{P_0G+B}{Rt}+\frac{at}{2}.    
\end{equation}

By the first order condition (FOC), $Tr_B(t)$ is minimized when



\begin{equation}\label{optimal_time}
    t_B := t =  \sqrt{\frac{2(P_0G+B)}{aR}}.
\end{equation}

Plugging $t_B$ obtained above in~\eqref{per_transaction_blob}, gives that the cost per transaction is

$$Tr_B(t) = \frac{P_0G+B}{R}\sqrt{\frac{aR}{2(P_0G+B)}}+\frac{1}{2}\sqrt{\frac{2(P_0G+B)a}{R}} = \sqrt{\frac{2(P_0G+B)a}{R}} = at_B,$$

and the number of transactions per blob is $C_B:=Rt_B = \sqrt{\frac{2(P_0G+B)R}{a}}.$

Next, we consider an {L1-posting strategy}.
Suppose a rollup posts a batch every time $t$, with $Rt$ transactions per batch.  The total cost of a batch is

\begin{equation}\label{total_batch_cost_t}
    T_E(t) := (P_0+RtP_1)G + \frac{aRt^2}{2},
\end{equation}

and the cost per transaction is

\begin{equation}\label{per_transaction_batch}
    Tr_E(t):=\frac{P_0G}{Rt}+P_1G+\frac{at}{2}.
\end{equation}

By the FOC, this is minimized when 
$t = \sqrt{2P_0G/(aR)}.$
Plugging the above value in~\eqref{per_transaction_batch} gives that the cost per transaction is:

$$Tr_E = \frac{P_0G}{R}\sqrt{\frac{aR}{2P_0G}}+P_1G+\frac{a}{2}\sqrt{\frac{2P_0G}{aR}} = \sqrt{\frac{2P_0Ga}{R}} + P_1G,$$ 

and the number of transactions per batch is $Rt = \sqrt{2P_0GR/a}.$

Now, we focus on the {indifference condition} between posting blobs and posting on L1.
That is, we solve for the value of $B$ that makes the rollup indifferent between the two strategies. For this, set costs per transaction in both cases equal: 

$$\sqrt{\frac{2(P_0G+B)a}{R}} = P_1G+\sqrt{\frac{2P_0Ga}{R}}.$$

Multiplying both sides by $\sqrt{R/(2a)}$ and squaring gives $P_0G+B = (\sqrt{R/2a}P_1G+\sqrt{P_0G})^2,$ 
or equivalently

$$B+P_0G = \frac{RP_1^2G^2}{2a}+2P_1G\sqrt{\frac{RP_0G}{2a}} + P_0G.$$

Canceling terms, we get

\begin{equation}\label{indifference_condition}
    B = \frac{RP_1^2G^2}{2a}+2P_1G\sqrt{\frac{RP_0G}{2a}}.
\end{equation}

For any finite $G$ and $B>0$, there are two scenarios. First, there is an index $n^*$ such that in~\eqref{indifference_condition}, the right-hand side is lower than the left-hand side, therefore, the rollup prefers to use a main market posting strategy, a contradiction. Second, if such an index does not exist, then we set the threshold equal to $n$. This finishes the proof of the proposition.  \qed

\end{proof}

Next, we demonstrate how to calculate the equilibrium price $B$ and the threshold in Proposition~\ref{threshold_result}.
For simplicity, suppose that $P_0=0$. 
Then, the time between posting for rollup $i$ is $t_i = \sqrt{\frac{2B}{aR_i}}$. The algorithm proceeds in two steps: 

 \begin{itemize}
     \item[$\bullet$] Step 1. Initialization:
        To hit the target of $k$ blobs per time unit, we require that
        $$k = \sum_{i=1}^{n} \frac{1}{t_i} = \sum_{i=0}^{n}\frac{\sqrt{aR_i}}{\sqrt{2B}}.$$
        Solving $B$ gives:
        \begin{equation}\label{price_formula}
        B = \frac{a(\sum_{i=1}^{n}\sqrt{R_i})^2}{2k^2}.            
        \end{equation}
        Let $\sum_{i=1}^{n}\sqrt{R_i}=:R$ is a positive real number. Plugging this value in~\eqref{indifference_condition}, gives the initial value on $m$ such that for a rollup with rate $R_m$, LHS of~\eqref{indifference_condition} is larger than RHS of~\eqref{indifference_condition}. 
    \item[$\bullet$] Step 2. 
    In the loop, we increase $m$ initialized in the previous step by one, calculate new equilibrium price $B$ with a set of rollups $\{1,2,...,n\}$, as long as the LHS of~\eqref{indifference_condition} is smaller than the RHS. Once we find a value of $k$, for which LHS is higher than the RHS, we output $B$ and $n^*=m$, as an answer.
\end{itemize}

Note that the condition in Step 2 may never be satisfied. In this case, $n^*$ is set to $n$. Intuitively, the calculation of $B$ in the initialization step assumes that all rollups use a blob posting strategy. However, it might be that some rollups under this price will not use a blob posting strategy, that is, it is an overestimation of the price. The second step fixes this potential overestimation by first excluding all small rollups and adding them one by one.

\begin{example}
Consider an example in which rates drop exponentially, that is, suppose $R_i=\frac{R_{i-1}}{2}$ for any $i\in \{1,...,n\}$. 

Assume that $G$ is very large, that is, all rollups use a blob posting strategy. To hit the target of $k$ blobs per time unit, we require that:

$$k = \sum_{i=0}^{n} \frac{1}{t_i} = \sqrt{\frac{aR_0}{2B}} \sum_{i=0}^{n} (\sqrt{2})^{-i} \approx \sqrt{\frac{aR_0}{2B}}\frac{\sqrt{2}}{\sqrt{2}-1} = \sqrt{\frac{aR_0}{B}}\frac{1}{\sqrt{2}-1}.$$

The approximation is taken by assuming a large enough value of $n$. We get an equivalent condition  $k(\sqrt{2}-1) \approx \sqrt{aR_0/B}.$ Solving for $B$ gives:

$$B \approx \frac{aR_0}{k^2(3-2\sqrt{2})}.$$

For $k=2$, an initial EIP-4844 target number of blobs per block, $B\approx 1.46aR_0,$ and  the time between posting for rollups is $1.71, 2.42, 3.42,...$
For $k=3$, a current EIP-4844 target, $B\approx 0.65aR_0,$ and  the time between posting for rollups is $1.14, 1.62, 2.27,...$

\end{example}

\section{Joining chains}

Suppose two rollups join forces in posting blobs. There are three different type of profiles of these rollups in the equilibrium derived above. In the first, both rollups use blob posting technology. In the second, one rollup uses blob posting technology, while the other uses the main market to post the data. In the third, both rollups use the main market for posting the data. In this section, we analyze what happens with the equilibrium price of the blob in these different scenarios and derive a cost-sharing scheme that satisfies certain reasonable properties. Let $B^{N}$ denote the new price in the equilibrium after two rollups join in posting blobs.

\medskip
\noindent \textbf{Case 1: } In this case, both rollups post the blobs in the initial equilibrium state. We obtain that the blob price in the new equilibrium state decreases, because of the blob price formula~\eqref{price_formula}. In the following, we show a result of how large this decrease can be.

\begin{proposition}\label{both_blob_posting}
    $B^{N}$ satisfies the following inequalities: $B\geq B^{N} \geq B/2$. 
\end{proposition}

\begin{proof}
    Assume that the two rollups joining in the blob posting are indexed  $i$ and $j$. Then, $B$ can be rewritten as $B=c(\sqrt{R_i}+\sqrt{R_j}+t)^2$ and $B^N=c(\sqrt{R_i+R_j}+t)^2$, where $c=\frac{1}{2k^2}$ and $t=\sum_{k\neq i,j}\sqrt{R_k}$. $B\geq B^N$ is equivalent to $\sqrt{R_i}+\sqrt{R_j}\geq \sqrt{R_i+R_j}$, that trivially holds for any $R_i,R_j>0$. The second inequality, $B^N\geq B^{N}/2$, is equivalent to $\sqrt{2}\sqrt{R_1+R_2}\geq \sqrt{R_1}+\sqrt{R_2}$. The latter is equivalent to $(R_1-R_2)^2\geq 0$, which holds trivially. 
\end{proof}

From the proof above we see that the equality in $B^N=B/2$ holds if and only if there are only two rollups and their transaction rates are equal. For obtaining the corner solution, it is also implicitly assumed that no other rollup joins the blob posting strategy, in the new state with a lower price. 

\medskip
\noindent \textbf{Case 2: } In this case, one rollup posts blobs, and the other posts on the main market in the initial arrangement. Joining blob posting pushes the price of the blob posting up, assuming that no rollup stops using blob posting technology. In the following, we derive an upper bound on the price increase.

\begin{proposition}\label{different_strategies}
        $B^{N}$ satisfies the following inequalities: $2B\geq B^{N}\geq B$. 
\end{proposition}

\begin{proof}
Assume that the rollup that posts blobs is indexed $i$, that is, its transaction rate is $R_i$ and the rollup that posts calldata at the main market has a transaction rate $R$. Then, the old blob price in the equilibrium is equal to $B=c(\sqrt{R_i}+t)^2$, where $c=\frac{1}{2k^2}$ and $t=\sum_{k\neq i}\sqrt{R_k}$. The new price, on the other hand, is equal to $B^N=c(\sqrt{R_i+R}+t)^2$. It is obvious that $B^N\geq B$, since no rollup stops using blob posting strategy. $2B\geq B^N$ is equivalent to $\sqrt{2}(\sqrt{R_i}+t) \geq \sqrt{R_i+R}+t$. It is sufficient to show that $\sqrt{2}\sqrt{R_i}\geq \sqrt{R_i+R}$, equivalent to $R_i\geq R$. The latter holds because the rollup with transaction rate $R_i$ posts blobs in the equilibrium and has a higher transaction rate than the rollup with $R$ rate, posting on the main market. 

\end{proof}

\medskip
\noindent \textbf{Case 3: } In this case, both rollups use the main market for posting the data in the initial setting. Assume they join in posting blobs and no rollup stops using blob posting technology. Then, the new blob price $B^N$ in the equilibrium increases. We obtain the following upper bound on the increase. 

\begin{proposition}\label{both_main_market}
    $B^N$ satisfies the following inequalities: $2B\geq B^{N}\geq B$. 
\end{proposition}

\begin{proof}
    The proof is similar to the proof of proposition~\ref{different_strategies}. Two rollups with transaction rates $R^1$ and $R^2$ that were posting their data in the main market, can reach a level that is almost $2R_1$. In fact, as long as $R^1+R^2\geq R_1$, the rollup with rate $R_1$ would not post blobs in the new equilibrium, a contradiction with the assumption. This gives an upper bound of $2B$ on the new equilibrium price. 
\end{proof}

\subsection{Cost Sharing}

Suppose there are two chains, with transaction arrival rates $R_L = R$, from now on referred to as large rollup, and $R_S = Rf$, referred to as small rollup. $0<f<1$ is a real number. Assume $P_0=0$, that is, there are no metadata costs.   
We use the same model as above. Let $Tr_L$ and $Tr_S$ denote costs per transaction of large and small rollups, respectively, $C_L$ and $C_S$ denote the total number of transactions posted by large and small rollups separately.

In the following, for the illustration of calculating these parameters above, in this and the next subsections, we assume that both rollups use a blob posting strategy. 
If the big chain is the only one using the blob space and it posts every $t_L$ time, then optimal posting time is $t_L = \sqrt{2B/(aR)},$
with a cost per transaction of $Tr_L=\sqrt{{2Ba}/{R}}$ and the total number of transactions in blob $C_L=\sqrt{2BR/a}$.

For the {small chain only}, we have optimal posting time  $t_S = \sqrt{2B/(aRf)},$
with a cost per transaction of  $Tr_S=\sqrt{{2Ba}/{(Rf)}}$ and the total number of transaction in the blob is $C_S=\sqrt{2BRf/a}$.

If the two chains post their blobs separately, the small chain has a higher per transaction cost, by a factor of $1/\sqrt{f}$ per transaction, since $t_S/t_R = 1/\sqrt{f}$.  Because the large chain has more transactions, it pays more overall, by a factor of $1/\sqrt{f}$.

A joint blob is also posted so that it minimizes cost per transaction. This is a Pareto efficient approach, as otherwise, both rollups could agree to deviate to the optimal strategy and share added value in any way. Based on the analysis in the proof of proposition, a joint blob is posted every
$t_J = \sqrt{\frac{2B^N}{(1+f)aR}}$.  The total cost per blob is:

$$B^N + \frac{a(1+f)R}{2}\cdot\frac{2B^N}{(1+f)aR} = 2B^N.$$

In all three cases, the total cost per blob is equal to $2$ times the blob price. 
A cost per transaction of  $Tr_J=\sqrt{{2B^Na}/{(R(1+f))}}$ and the total number of transaction in the blob is $C_J=\sqrt{2B^NR(1+f)/a}$.

First, note that if $Tr_J>Tr_L$, the rollups will not join in posting blobs together, as it is not profitable for a large rollup. Therefore, an interesting case is when $Tr_J<Tr_L\leq Tr_S$. In propositions~\ref{both_blob_posting},~\ref{different_strategies} and~\ref{both_main_market}, we obtained that $2B\geq B^N\geq B/2$. That is, the relation between $Tr_J$ and $Tr_L$ can be arbitrary. We discuss a suitable cost-sharing rule in the next section.

\subsection{The Nash Bargaining Solution}

In this section, we take an axiomatic approach to the cost-sharing rule between rollups that decide to post blobs together. One such approach is suggested by the Nash Bargaining solution. First, we introduce the required notation and then reduce the cost-sharing rule to solving the Nash Bargaining problem. Let $A$ denote the set of all possible bargaining outcomes. In particular, $D\in A$ is the outcome if no agreement can be reached. In our setting, $A$ is interpreted as a set of cost-sharing options between two rollups, while $D$ is interpreted as a case when rollups post their blobs separately.
The utility (payoff) function of agent $i$ is given by 
$u_i: A \rightarrow \mathbb{R}.$
We consider linear utility functions, in particular. Let $S$ denote the set of all possible utilities (payoffs):
$$S=\left\lbrace (s_1,s_2)\mid s_1=u_1(a),s_2=u_2(a),a\in A\right\rbrace.$$
Let $s:=(s_1,s_2)$. 
 Further, let $d=(d_1,d_2)=\left( u_1\left( D\right),u_2\left( D\right) \right)$ be the utility vector if no agreement could be reached   (threat point). 
Two Requirements on $S$ to have a characterization: 1) There exists $s\in S$ with $s_i>d_i$  $\forall i$, and 2) $S$ is compact and convex.
Our set satisfies these properties, as we will see later. 
Then, define
    $S'=\left\lbrace s\mid s_i \geq d_i \hspace{0.2cm}\forall 
      i\right\rbrace \subseteq S.$
Let $H(S)$ denote the set of {\em Pareto-optimal} outcomes. {\em Pareto-optimal} (or {\em Pareto-efficient}) means that there is no other outcome that makes one player better off without making another player worse off. In our case, this means that rollups pay completely for the blob posting cost and do not overpay. Now, we are in the position to define a Nash Bargaining Solution.

\begin{definition}
    A {\bf bargaining solution} is a rule that assigns a solution
    vector $f(S,d)\in S$ 
    to every bargaining problem $B=(S,d)$.
  \end{definition}

Let $f_i(S,d)$ denote the $i$-component of $f(S,d)$. 
    That is:
    $f(S,d)=\left(f_1 \left(S,d\right),f_2 \left(S,d\right)\right).$
We have the following 4 axioms. 

\begin{axiom}[\bf Invariance of Utility Scaling]

If there are two bargaining situations 
        $B=(S,d)$ and $\bar{B}=(\bar{S},\bar{d})$ with
        $\bar{S}=\{\alpha_1 s_1 + \beta_1, \alpha_2 s_2 + \beta_2: \,\, 
        s\in S\}$
        and $\bar{d}_i=\alpha_i d_i+ \beta_i \,\,\, \forall i,$
        where $\alpha_1,\alpha_2>0$.
      Then, for the solution, the following holds:
             $f_i(\bar{S},\bar{d})=\alpha_i f_i(S,d)+\beta_i \,\,\,\forall i.$ 
\end{axiom}

The axiom states that if we change the way we measure utility when we construct a bargaining problem but keep new
  utility scales decision-theoretically equivalent to the old ones, then the bargaining solution in utility-allocation space
  changes in the same way, so that it still corresponds to the same real outcome.
  
\begin{axiom}[\bf Pareto Optimality]
If $f(S,d)$ is a solution to $B=(S,d)$, then $f(S,d)\in H(S)$.
\end{axiom}

The axiom states that there is no other feasible allocation that is better than the solution for one player and not worse than the solution
  for the other player. For the next axiom, we need a definition: 

\begin{definition}
    A game is called symmetric if two conditions hold: (1) $d_1=d_2$, and (2) $(s_1,s_2)\in S$, then $(s_2,s_1)\in S$.
\end{definition}

\begin{axiom}[\bf Symmetry]
If $(S,d)$ is symmetric, then the solution is also symmetric, i.e. $f_1(S,d)=f_2(S,d).$    
\end{axiom}
      
The axiom states that, if the positions of players $1$ and $2$ are completely symmetric in the bargaining problem, then the solution also treats them symmetrically.

\begin{axiom}[\bf Independence of Irrelevant Alternatives]
Let $(S,d)$ and $(T,d)$ be two bargaining situations with $S\subset T$ and $f(T,d)\in S$, then $f(T,d)=f(S,d)$.
\end{axiom}
 
The axiom states that eliminating feasible alternatives (other than the threat point) that would not have been chosen does not affect the result.

\begin{theorem}[\cite{nash1950bargaining}]
 There is a unique bargaining solution, $f^N$, satisfying the 
 four axioms above and it has the 
 following representation for every two-person bargaining problem:
 \begin{equation}\label{Nash_optimization}
  f^N(S,d)=\arg \max_{s\in H(S)}(s_1-d_1)(s_2-d_2)=s^*.     
 \end{equation}
\end{theorem}

The expression $(s_1-d_1)(s_2-d_2)$ is called the Nash Product. Suppose the large rollup is indexed by $1$ and the small rollup is indexed by $2$. The disagreement point in our setting is $(d_1,d_2) = (Tr_L, Tr_S)$. Assuming that posting a joint blob is profitable, that is, $Tr_J<Tr_L$, the rollups need to decide how to share the new blob price $B^N$. Suppose the large rollup pays $B_1$ and the small rollup pays $B_2$, with $B_1,B_2\geq 0$. Then, we can redefine their per-transaction costs, which define points $(s_1,s_2)$ in the payoff set. This defines a two-dimensional space. However, note that because of the Pareto efficiency, $B_1+B_2=B$ holds, since underpaying is not an option, and overpaying is not efficient. Therefore, we are down to $1$ dimensional space, as $B_1\in [0,B^N]$ defines it fully.

The number of large and small rollup transactions in the blob are denoted by $C_{J,L}$ and $C_{J,S}$, respectively. Since their rate ratio is $\frac{1}{f}$, we have $C_{J,L}=\frac{C_J}{1+f}$ and $C_{J,S} = \frac{C_J f}{1+f}$, so that they sum up to $C_S$. Let $d_{J,L}$ and $d_{J,S}$ denote the total delay costs of the large and small rollups, respectively. Then $d_{J,L}=\frac{aRt_J^2}{2}$ and $d_{J,S}=\frac{aRft_J^2}{2}$. Having settled all necessary parameters, we proceed to calculate $s_1$ and $s_2$ values for given $B_1$. $s_1$ is calculated as $s_1 = \frac{B_1+d_{J,L}}{C_{J,L}}$ and $s_2$ is calculated as $s_2=\frac{B^N-B_1+d_{J,S}}{C_{J,S}}$. Plugging in $s_1$ and $s_2$ in~\eqref{Nash_optimization}, and simplification by getting rid of constant denominators gives the following optimization problem:

\begin{equation}\label{optimize_payment}
 \arg\max_{B_1}(B_1+d_{J,L}-C_{J,L}Tr_L)(B^N-B_1+d_{J,S}-C_{J,S}Tr_S).   
\end{equation}

Since~\eqref{optimal_payment} is a negative quadratic polynomial in $B_1$, we solve the optimal value by the first order condition with respect to $B_1$:

\begin{equation}\label{optimal_payment}
    B_1= (B^N+d_{J,S}-d_{J,L}-C_{J,S}Tr_S+C_{J,L}Tr_{L})/2.
\end{equation}

The solution of~\eqref{optimal_payment} directly gives a sharing rule for a blob cost and also determines per transaction costs for both rollups. Note that it is only a function of $f$, and therefore, the contract between the rollups can be easily automatized. 
The axiomatic approach of this section can be easily generalized to $m>2$ rollups, by taking a Nash product over $m$ rollups $(s_1-d_1)(s_2-d_2)\cdot...\cdot (s_m-d_m)$. However, the optimization problem at hand can be much harder to solve, as it is $m-1$ dimensional. The number of dimensions comes from $m-1$ rollups' contributions towards the final blob cost. The last rollup contribution is determined by the contributions of the rest.

First, we show a structural result that will come in handy later. The result is similar to the one in proposition~\ref{both_blob_posting}, in that we lower bound the new equilibrium price in terms of the original equilibrium price $B$, and a parameter $f$.

\begin{lemma}\label{structural}
    $\frac{B^N}{B}\geq \frac{1+f}{(1+\sqrt{f})^2}$.
\end{lemma}

\begin{proof}
    Similar to the proof of proposition~\ref{both_blob_posting}, the ratio between $B^N$ and $B$ is minimized if there are only two rollups posting blobs. Then, in this case, by~\eqref{price_formula} we have that $B^N/B = \frac{\sqrt{R+Rf}}{(R+\sqrt{Rf})^2}=\frac{1+f}{(1+\sqrt{f})^2}$, which finishes the proof. 
\end{proof}

Next, we show that the Nash bargaining outcome the large rollup to pay, is always an internal value:

\begin{proposition}\label{internal_solution}
    $B_1\in (0,B^N)$ for any $f<1$.
\end{proposition}

\begin{proof}
Plugging all values in the formula of $B_1$ gives: 

\begin{align*}
B_1 = & 0.5(B^N + \frac{aRf}{2}\frac{2B^N}{a(1+f)R} - \frac{aR}{2} \frac{2B^N}{a(1+f)R} - \\ & \frac{f}{1+f}\sqrt{\frac{2B^N(1+f)R}{a}}\sqrt{\frac{2Ba}{Rf}} +
\frac{1}{f+1}\sqrt{\frac{2B^N(1+f)R}{a}}\sqrt{\frac{2Ba}{R}}).    
\end{align*}

Simplifying gives: 

\begin{equation}\label{simplified_B_1}
    B_1 = (B^N\frac{f}{1+f} + \sqrt{B^NB}(\sqrt{\frac{1}{1+f}} - \sqrt{\frac{f}{1+f}})). 
\end{equation}

Then, $B_1 < B^N$ is equivalent to: 

\begin{equation*}
    B(1+f)(1-\sqrt{f})^2 < B^N.
\end{equation*}
Note that the condition in the lemma~\ref{structural} readily implies the condition above, which finishes the proof of the proposition.
\end{proof}

Note that $B_1$ does not depend on $R$ and $a$, see~\eqref{simplified_B_1}, but only on blob prices and $f$, as claimed earlier. 
To get an intuition of the parameters above, in the following, we consider an example. Let $Tr_{J,L,B_1}$ and $Tr_{J,S,B_1}$ denote effective per transaction costs after fixing the amount large rollup pays for the blob price, $B_1$. The following holds:

\begin{equation}
    Tr_{J,L,B_1}=(B_1+d_{J,L})/C_{J,L} \text{ and } Tr_{J,S,B_1}=(B^N-B_1+d_{J,S})/C_{J,S}.
\end{equation} 

Let the rollup $X\in \{L,S\}$ improvement is denoted by $$I_{X,B_1}:=(Tr_X-Tr_{J,X,B_1}) / Tr_X = 1 - Tr_{J,X,B_1}/Tr_X,$$ and the proportional payment of the large rollup $B_1^{pr}:=\frac{B^N}{1+f}$.

\begin{example}
    Suppose $R=B=a=1$ and $f=0.25$. That is, a small rollup has 4 times less traffic than a large rollup.  Then, in the case of large rollup posting blobs alone, parameters are equal  $t_L=Tr_L = C_L = \sqrt{2}\approx 1.41$. Parameters of the small rollup posting alone:  $t_S=Tr_S = \sqrt{8} \approx 2.82$ and $C_L=\sqrt{0.5}\approx 0.71$. The joint posting parameters are: $t_J=Tr_J=\sqrt{2\cdot 0.81/1.25}\approx 1.14$ and $C_J=\sqrt{2.5\cdot0.81}\approx 1.42.$ 
    Large rollup includes $C_{J,L} = \frac{1.42}{1.25}\approx 1.138$ transactions in the joint blob, small rollup includes $C_{J,S}\approx 0.285$ transactions. Large rollup total delay is $d_{J,L}\approx 1.14^2/2 = 0.65$ and small rollup delay is $d_{J,S}=0.25d_{J,L}\approx 0.16.$ 
    Finally, we plug in all parameters in the calculation of the large rollup share in the Nash bargaining solution~\eqref{optimal_payment}: 
    $$B_1\approx (0.81+0.16-0.65-0.285\cdot 2.82 + 1.138\cdot 1.41)/2\approx 0.564.$$ Then, the small rollup pays $B_2\approx 0.15$. Note that they do not share the total price $0.81$ proportionally, which would result in the large rollup paying $B_1^{pr}=0.81\cdot \frac{4}{5}=0.648$.  Plugging in all values, we obtain $Tr_{J,L,B_1}=(0.56+0.65)/1.138\approx 1.07$ and $Tr_{J,S,B_1} = (0.15+0.16)/0.28\approx 1.43.$ That is, the large rollup improvement is $I_{L,B_1}\approx 24.7\%$, while the small rollup improvement is $I_{S,B_1} \approx 49.4\%$. Note that the large rollup improvement of the per-transaction cost is smaller than the improvement of the small rollup. 
\end{example}

Observations obtained in the example above are more general, which we show in the following propositions.
First, we show that the Nash bargaining outcome for the large rollup is less than the proportional payment to the blob price: 

\begin{proposition}\label{nash_vs_proportional}
    $B_1\leq B_1^{pr}$ for any $f<1$.
\end{proposition}

\begin{proof}
$B_1<B^{pr} = \frac{B^N}{1+f}$ from~\eqref{simplified_B_1} is equivalent to 

\begin{equation}
\sqrt{\frac{B^NB}{1+f}}(1-\sqrt{f}) < \frac{B^N(1-f)}{1+f},    
\end{equation}

which is on its own equivalent to $B<\frac{(1+\sqrt{f})^2}{1+f}B^N$. This condition is exactly the condition in the lemma~\ref{structural}, which finishes the proof of the proposition.

\end{proof}


Next, we obtain a lower bound on the Nash bargaining outcome. Namely, we show the following:

\begin{proposition}\label{lower_bound_B1}
    $B_1\geq B^N/2$ for any $f<1$. 
\end{proposition}
\begin{proof}
    The condition $B_1\geq B^N/2$ is equivalent to 
    
    $$\sqrt{B^N B}\frac{1-\sqrt{f}}{\sqrt{1+f}} \geq \frac{1-f}{2(1+f)}B^N, $$ which after simplification becomes:

    $$\frac{B^N}{B}\leq 4(1+f)(1+\sqrt{f})^2.$$

    The right-hand side of the above inequality is decreasing in $f$ and achieves its minimum value $2$ when $f=1$. By proposition~\ref{both_blob_posting}, we know that $B\geq B^N$, which finishes the proof.  
\end{proof}

That is, the large rollup never pays less than half of the new blob price, which is fair. 

Last, we show that the large rollup improvement in the Nash bargaining outcome is less than the small rollup improvement. 

\begin{proposition}\label{nash_vs_equal_improvement}
    $I_{L,B_1}\leq I_{S,B_1}$ for any $f<1$.
\end{proposition}

\begin{proof}
    The condition is equivalent to: 

    \begin{equation}
        B^N\frac{(1+\sqrt{f})f}{1+f} +\sqrt{B^NB}(1+\sqrt{f})(\sqrt{\frac{1}{1+f}}-\sqrt{\frac{f}{1+f}})\geq \frac{B^N}{1+f}(f\sqrt{f}-1)+B^N\sqrt{f}. 
    \end{equation}

    Further simplification gives: 

    \begin{equation*}
        \frac{2}{1+f}+\sqrt{\frac{B}{B^N}}\frac{1-f}{\sqrt{f+1}} - B^N\sqrt{f}\geq 0.
    \end{equation*}

    Let $p:=\sqrt{B/B^N}$. We show that the function 
    $$h(f):=\frac{2}{1+f}+p\frac{1-f}{\sqrt{1+f}}-\sqrt{f}$$ 

    is decreasing in $f$ on the interval $[0,1]$ for any $p>0$. Note that $dh(f)/df<0$ for any $f>0$. Since $h(1)=0,$ we get the proof of the proposition.
\end{proof}

The intuition is simple: a small rollup has much more room to improve than a large rollup. Note that the result holds {\it unconditionally} regarding the equilibrium prices $B$ and $B^N$. Since we consider only Pareto-efficient solutions, the result, in particular, implies that a Nash bargaining solution favors the small rollup compared to the ''fair" blob cost-sharing rule, which improves both rollup per-transaction costs equally. The latter favors the large rollup to a high extent: the rollups pay almost equally even if $f=0.25$.   

In this section, we assumed that the two rollups engaged in the joint blob posting were the ones that initially were posting blobs. This in particular gives a lower bound on the new equilibrium price, derived in the Lemma~\ref{structural}, and is used in the proofs of propositions~\ref{internal_solution} and~\ref{nash_vs_proportional}. If one or both rollups were using the main market, then the propositions would be automatically satisfied, as the new equilibrium price goes up.   

\section{Extensions}

In this section, we discuss two natural extensions of the baseline model. In the first extension, the blob size is limited. Suppose there is the maximum blob size $U$, so that for any rollup, $Rt \leq U$, or equivalently, the posting time $t$ is upper bounded by  ${U/R}$. Intuitively, adding an upper bound on the blob size causes the blob price in the equilibrium to increase, compared to the baseline model. The reason is that with the upper bound the rollups produce blobs even faster.  
To figure out rollup optimum posting time in the case of using the data market, we again solve the optimal time of posting using the first-order condition and compare the blob size with the upper bound $U$. If the obtained size is larger than the upper bound, then we take the size to be $U$ and adjust the posting time accordingly. If on the other hand, the optimal size is smaller than the upper bound, the rollup keeps the same optimal strategy.
Calculating the main market posting time is done exactly as in the main model, therefore, deciding on the posting strategy given equilibrium price is trivial. Calculating the equilibrium price is also easy. 

Shared blob posting stays the same if the aggregate demand does not cross the threshold, and cost-sharing does not need modification as well. If one of the participating rollups reaches the threshold size itself and the other one does not, then it has more bargaining power, as it saves only on the delay cost. However, if both rollups reach the threshold themselves, cost-sharing becomes more intricate. Both rollups save only on the delay cost, and therefore, large rollup has lower bargaining power over a small rollup, compared to the baseline model.

In the second extension,   rollups have compression technology. It is natural to assume that compression technology is monotonic, that is, the compression factor is increasing in the data size. Then, the existence of compression, in general, favors bigger rollups as well, since they generate enough transaction data to compress efficiently faster than smaller rollups do. In the case of shared blob posting, this advantage should be compensated. The Nash bargaining outcome guarantees such compensation since the per-transaction cost in the disagreement point for the small rollup will be low because of compression, as the joint blob posting does not let the small rollup use compression to full extent.

\section{Conclusions and Future Work}
We introduced a simple economic model to analyze EIP-4844 and its effects on rollups to decrease L1 data costs. In the proposed model, large enough rollups use a new market for posting their data to L1, while the rest continue using the original market. Moreover, we studied sharing blob posting and cost-sharing rules. First, we outlined conditions when sharing is profitable for both rollups and then described an axiomatic approach to the blob-posting cost-sharing rule.  
The are many interesting future research avenues: (a) optimal strategy in an oligopolistic market with relevant strategic interaction between rollups;  (b) strategic consumers and endogenous main market equilibrium price --
we can extend the above model with demand growth to study questions about the equilibrium structure of demand;
(c) finally, using agent-based simulation, we can numerically test the theoretical results obtained above in an environment that closely represents the actual Ethereum market and proposed data market with dynamic transaction fee adjustment. For example, we can model discrete block time and compression technology of rollups. This would allow us to validate our results and provide insights into practical rollup posting policies/services as well as L1 fee market design.

\bibliographystyle{acm}
\bibliography{name}

\end{document}